\RequirePackage{amsmath}
\documentclass[runningheads]{amsart}
\usepackage[foot]{amsaddr}

\usepackage{color}
\usepackage{amssymb, amsthm}
\usepackage{graphicx}

\usepackage{tikz}
\usetikzlibrary{arrows,automata,positioning,decorations.pathreplacing,calc,intersections}

\usepackage[colorlinks=true,
linkcolor=webgreen,
filecolor=webbrown,
citecolor=webgreen]{hyperref}

\definecolor{webgreen}{rgb}{0,.5,0}
\definecolor{webbrown}{rgb}{.6,0,0}
\definecolor{myred}{RGB}{0,0,0}

\usepackage{float}

\usepackage{amsmath,amssymb}
\usepackage{amsfonts}
\usepackage{latexsym}
\usepackage{epsf}
\usepackage{tikz}

\DeclareMathOperator{\scc}{sc}

\newcommand{\seqnum}[1]{\href{https://oeis.org/#1}{\underline{#1}}}

\newcommand{\abs}[1]{\left| #1 \right|}
\newcommand{\MA}{{\mathcal A}}
\newcommand{\mn}[1]{\overline{#1}}

\theoremstyle{plain}
\newtheorem{thm}{Theorem}
\newtheorem{lem}[thm]{Lemma}
\newtheorem{prop}[thm]{Proposition}
\newtheorem{cor}{Corollary}

\theoremstyle{definition}
\newtheorem{defn}{Definition}

\theoremstyle{remark}
\newtheorem*{rem}{Remark}

\begin{document}

\title{Maximal State Complexity and Generalized de Bruijn Words}

\author{Daniel Gabric}
\email[A1]{dgabric@uwaterloo.ca} 

\author{\v{S}t\v{e}p{\'a}n Holub}
\email[A2]{holub@karlin.mff.cuni.cz}

\author{Jeffrey Shallit}
\email[A3]{shallit@uwaterloo.ca}
\address[A1,A3]{
School of Computer Science,
University of Waterloo,
Waterloo, Ontario N2L 3G1,
Canada}
\address[A2]{	
	Department of Algebra,
	Faculty of Mathematics and Physics,
	Charles University,
	Prague,
	Czech Republic}

\begin{abstract}
	We compute the exact maximum state complexity for the language
	consisting of $m$ words of length $N$, and characterize languages achieving the maximum. We also consider a special case, namely languages $C(w)$ consisting of the conjugates of a single word $w$. The words for which the maximum state complexity of $C(w)$ is achieved turn out to be a natural generalization of de Bruijn words. We show that generalized de Bruijn words exist for each length and consider the number of them.
\end{abstract}

\maketitle

\section{Introduction}

Let $x, y$ be words.  We say $x$ and $y$ are {\it conjugates\/} if one is a 
cyclic shift of the other; equivalently if there exist words $u,v$
such that $x = uv$ and $y = vu$.  For example, the English words
{\tt listen} and {\tt enlist} are conjugates.

The set of all conjugates of a word $w$ is denoted by $C(w)$.  Thus,
for example, $C({\tt eat}) = \{ {\tt eat, tea, ate} \}$.  
We also write $C(L)$ for the set of all conjugates of elements
of the language $L$.

For a regular language $L$ let $\scc(L)$ denote the {\it state complexity\/}
of $L$:  the number of states in the smallest complete DFA accepting
$L$.   State complexity is sometimes also called {\it quotient complexity}
\cite{Brzozowski:2010}.   The state complexity of the cyclic shift operation
$L \rightarrow C(L)$ for arbitrary regular languages
$L$ was studied in Maslov's pioneering 1970 paper \cite{Maslov:1970}.
More recently, 
Jir\'askov\'a and Okhotin \cite{Jiraskova&Okhotin:2008} improved Maslov's bound,
and Jir\'asek and Jir\'askov\'a
studied the state complexity of the conjugates of prefix-free languages
\cite{Jirasek&Jiraskova:2013}.

In this note we investigate the state complexity of \emph{uniform-length} languages, that is, of languages $L\subseteq \Sigma^N$. The language $C(w)$ is a special case of a uniform-length language.   Clearly $\scc(C(x))$ achieves its minimum --- namely, $N+2$ --- at words of the form $a^N$, where $a$ is a single letter.  
By considering random words, it seems likely that
$\scc(C(x)) = \Theta(N^2)$ in the worst case.

In Theorem~\ref{thm_complexity}, we prove an exact bound   for the state complexity of (almost all) uniform-length languages and characterize languages that attain the bound. In particular, this means that we determine the state complexity of cyclic shift on languages consisting of a single word. Moreover, the characterization of words $w$ for which $C(w)$ achieves the maximum turns out to be a natural generalization of de Bruijn words to words of arbitrary length. Therefore, in Section \ref{gendeb}, we introduce the concept of \emph{generalized de Bruijn word} and show that such words exist for each length. 

This paper is the journal version of the conference
paper \cite{Gabric&Holub&Shallit:2019}.  It differs in several respects
from that paper:  we have reworked the discussion of the necessary
concepts from graph theory (in Section~\ref{gendeb}), providing more
details; we have characterized uniform-length languages achieving maximum state complexity in Theorem~\ref{thm_complexity} which includes a corrected statement of Theorem 3 of the conference paper; and we have provided
additional enumeration details in Tables 1 and 4.

\section{Generalized de Bruijn words}
\label{gendeb}

De Bruijn words (also called de Bruijn sequences) have a long history \cite{Flye.Sainte-Marie:1894,Martin:1934,Good:1946,deBruijn:1946,deBruijn:1975}, and have been extremely
well studied \cite{Fredricksen:1982,Ralston:1982}.  Let $\Sigma_k$ denote the $k$-letter alphabet
$\{ 0, 1, \ldots, k-1 \}$.  
Traditionally, there are two distinct ways of thinking about these
words:  for integers $k \geq 2$, $n \geq 1$ they are
\begin{itemize}
\item[(a)] the words $w$  
having each word of length $n$
over $\Sigma_k$ exactly once as a factor; or

\item[(b)] the words $w$ 
having each word of length $n$
over $\Sigma_k$ exactly once as a factor, when
$w$ is considered as a ``circular word'', or ``necklace'', 
 where the word ``wraps around'' at the end back to the beginning.
\end{itemize}

For example, for $k = 2$ and $n = 4$, the word
$$ 0000111101100101000$$
is an example of the first interpretation and
$$ 0000111101100101$$ 
is an example of the second.

In this paper, we are concerned with the second (circular)
interpretation of de Bruijn words.  
Obviously, 
such words exist only for lengths of the form $k^n$.
Is there a sensible way to generalize this class of words
so that one could speak fruitfully of (generalized) de Bruijn words
of every length?

One natural way to do so is to use the notion of {\it subword complexity} 
(also called {\it factor complexity\/} or just {\it complexity}).
For $0 \leq i \leq N$ let
$\gamma_i(w)$ denote the number of distinct length-$i$ factors of the word
$w \in\Sigma_k^N$ (considered circularly).   For all words $w$, there is a natural upper bound 
on $\gamma_i (w)$ for $0 \leq i \leq N$, as follows:
\begin{equation}
\gamma_i (w) \leq \min(k^i, N).
\label{bnd}
\end{equation}
This is immediate, since there are at most $k^i$ words of length $i$
over $\Sigma_k$, and there are at most $N$ positions where a word
could begin in $w$ (considered circularly).

Ordinary de Bruijn words are then precisely those words $w$ of length $k^n$ for which
$\gamma_n (w) = k^n$.  But even more is true: a de Bruijn word  $w$ also
attains the upper bound in \eqref{bnd} for {\it all\/} $i \leq k^n$.
To see this, note that if $i \leq n$, then every word of length $i$
occurs as a prefix of some word of length $n$, and every word of
length $n$ is guaranteed to appear in $w$.   On the other hand,
all $k^n$ (circular) factors of each length $i \geq n$ are distinct, because their
length-$n$ prefixes are all distinct.

This motivates the following
definition:
\begin{defn}
A word $w$ of length $N$ over a $k$-letter alphabet is said to be a 
{\it generalized de Bruijn word} if $\gamma_i (w) = \min(k^i, N)$ for all $0 \leq i \leq N$.
\label{gubdef}
\end{defn}

Table~\ref{tab1} gives the lexicographically least de Bruijn words
for a two-letter alphabet, for lengths $1$ to $31$, and the number
of such words (counted up to cyclic shift).  This
forms sequence \seqnum{A317586} in the 
{\it On-Line Encyclopedia of Integer Sequences} (OEIS)
\cite{Sloane}.  

We point out an alternative characterization of our generalized
de Bruijn words.   

\begin{prop}\label{prop1}
A word $w \in \Sigma_k^N$ is a generalized de Bruijn word iff both
of the following hold:
\begin{enumerate}
\item  $\gamma_r (w) = k^r$; and \label{a}
\item  $\gamma_{r+1} (w) = N$, \label{b}
\end{enumerate}
where $r = \lfloor \log_k N \rfloor$.
\label{proptwo}
\end{prop}

\begin{proof}
A generalized de Bruijn word trivially has these properties. An argument similar to the discussion before Definition \ref{gubdef} shows that the two properties imply the equality $\gamma_i (w) = \min(k^i, N)$ for all $0 \leq i \leq N$.
\end{proof}

We now show that generalized de Bruijn words exist. Since one of the most powerful tools for studying de Bruijn words are de Bruijn graphs, we shall need some results from (directed) graph theory. Let us first set the terminology. A closed sequence of edges 
 $(v_0,v_1)$, $(v_1,v_2)$, $(v_2,v_3)$, \ldots , $(v_{n-1}, v_0)$ is called a \emph{cycle} (of length $n$) if all vertices $v_0, v_1, \dots, v_{n-1}$ are distinct. If all edges in the sequence are distinct (vertices may repeat), the sequence is called a \emph{circuit} (of length $n$).
A cycle that visits all vertices of a graph is called a \emph{Hamiltonian cycle}. 
A circuit traversing all edges is an \emph{Eulerian circuit}. A directed graph is an \emph{Eulerian graph} if, for each vertex $v$, the number of edges incoming to $v$ is the same as the number
of edges outgoing from $v$. It is well known that each connected component of an Eulerian graph admits an Eulerian circuit. If, for all vertices, the number of incoming edges, as well as the number of outgoing edges is $k$, then
the graph is said to be \emph{regular of degree $2k$}. The \emph{degree} of a vertex is the total number of its incoming and outgoing edges.

A \emph{factor} (more precisely a $2$-\/factor) of a graph is the set of 
vertex-disjoint cycles that together cover all vertices. Note, for example, that a Hamiltonian cycle is a special case of a factor. One of the first published results in graph theory is the following fact, proved in \cite[Claim 9, p.~200]{petersen1891}. (For a more contemporary proof,
see, for example, \cite[Theorem 3.3.9, p.~140]{graphintro}.)
\begin{lem}[Petersen]\label{petersen}
	Let $G$ be a regular graph of degree $2k$. Then the edges of $G$ can be partitioned into $k$ distinct factors.
\end{lem} 

 The \emph{$k$-ary de Bruijn graph of order $n$}, denoted $G_n^k$, is a directed graph where the vertices are the $k$-ary words of length $n$, and edges join a word $x$ to a word $y$ if
$x = at$ and $y = tb$ for some letters $a, b$ and a word $t$. 
An ordinary de Bruijn word $a_0a_1\cdots a_{k^n-1}$ of length $k^n$ can be represented by the cycle
 $(v_0,v_1)$, $(v_1,v_2)$, $(v_2,v_3)$, \ldots , $(v_{k^n-1}, v_0)$ where $v_i = a_ia_{i+1}\cdots a_{i+n-1}$, indices taken modulo $k^n$.
This establishes a one-to-one correspondence between Hamiltonian cycles of $G_n^k$ and de Bruijn words of length $k^n$. Similarly, there is a one-to-one correspondence between such words and Eulerian circuits in $G_{n-1}^k$ of the form 
$(v'_0,v'_1)$, $(v'_1,v'_2)$, $(v'_2,v'_3)$, \ldots , $(v'_{k^n-1}, v'_0)$ where $v_i'=a_ia_{i+1}\cdots a_{i+n-2}$, indices again taken modulo $k^n$. 
 More generally, edges in $G_{n-1}^k$ are in one-to-one correspondence with vertices of $G_n^k$, where the edge $(at,tb)$ corresponds to the vertex $atb$. Circuits in $G_{n-1}^k$ then correspond to cycles in $G_n^{k}$.

Every vertex of $G_n^k$ has $k$ incoming edges, and $k$ outgoing edges, and therefore $G_n^k$ is a regular graph of degree $2k$. The fact that such a graph is Eulerian yields the existence of ordinary de Bruijn words. By Proposition \ref{prop1}, it also becomes clear that  building a generalized de Bruijn word of length $N= k^n + j$, where $0\leq j\leq (k-1)k^n$, over a $k$-letter alphabet amounts to constructing a circuit of length $N$ in $G_n^k$ that visits every vertex.

The existence of generalized de Bruijn words of any length is almost proved in a paper by
Lempel \cite{Lempel:1971}.   Lempel proved that for all
$k \geq 2$, $n \geq 1$, $N \leq k^{n+1}$,
there exists a circular word $w = w(k,n,N)$ of length $N$ 
for which the factors of size $n$ are distinct.  
(Also see \cite{Hemmati&Costello:1978,Etzion:1986}.) In other words, Lempel shows the existence of a connected Eulerian graph with $N$ edges in $G_n^k$.
However, his proof does not explicitly state
that the circuit visits all vertices if $k^n \leq N$.  
The resulting word therefore satisfies condition \eqref{b}
of Proposition \ref{prop1}, but not necessarily condition \eqref{a}.
For example, the binary word $10011110000$ of length $11$ contains $11$ distinct circular factors of length $4$, but only $7$ factors of length $3$: the factor $101$ is missing (see Figure~\ref{path}).

A further analysis of Lempel's construction nevertheless reveals that this additional required property is satisfied.
For sake of completeness, we reconstruct the proof below. In fact, our proof more closely follows the proof by Yoeli \cite{Yoeli:1962} for the binary case, which, in turn, was followed by Lempel.
(A similar analysis of Yoeli's proof
in the binary setting can be found in \cite{Shallit:1993}.)

 The core of the proof are the following two facts about de Bruijn graphs.

\begin{lem}\label{complete}
  Let $k\geq 2$ and $n\geq 1$.	Then every cycle in $G_n^k$ can be completed to a factor.
\end{lem}
\begin{proof}
	For $n=1$, the graph $G_1^k$ contains a loop, i.e., the edge $(a,a)$, for each vertex $a$ where $a$ is a letter. A given cycle $C$ can be therefore completed with loops in vertices that are not contained in $C$. 
	
	Let $n \geq 2$ and let $C$ be a cycle in $G_n^k$. Consider the complement $H$ of the connected Eulerian graph corresponding to $C$ in $G_{n-1}^k$. The graph $H$ is Eulerian, and the cycles in $G_n^k$ corresponding to Eulerian circuits of connected components of $H$ together with $C$ form a factor of $G_n^k$.  
\end{proof}	

\begin{figure}[h]
\begin{center}
\begin{tikzpicture}
[
scale =1.0,
>=stealth, 
node distance = 8ex, inner sep = 2pt,
every node/.style = {font = \scriptsize}
]
\node (000) at (0,0) {$000$};
\node  (001)  at (1.5,-1) {$001$};
\node (100)  at (-1.5,-1) {$100$};
\node (010) at (0,-2) {$010$};
\node[fill = myred!20] (101) at (0,-3) {$101$};
\node (011) at (1.5,-4) {$011$};
\node (111) at (0,-5) {$111$};
\node (110) at (-1.5,-4) {$110$};
\path (000) edge[out = 130, in = 50, loop , draw = white] node[above] {{\bf \tiny 0}} (000); 
\path (111) edge[out = 310, in = 230, loop, draw = white] node[below] {{\bf \tiny 1}} (111);
\foreach \from/\to/\p in {000/001/1,001/010/0,100/001/1,100/000/0,101/011/1,011/111/1,011/110/0,111/110/0,110/101/1}
\path[->] (\from) edge node[above] {{\bf \tiny  \p}} (\to);
\foreach \from/\to/\p in {001/011/1,110/100/0} 
\path[->] (\from) edge node[left] {{\bf \tiny \p}} (\to); 
\path[->] (010.310) edge node[right] {{\bf \tiny 1}} (101.50); 
\path[->] (101.130) edge node[left] {{\bf \tiny 0}} (010.230); 
\foreach \from/\to in {000/001,001/010,100/001,001/011,011/111,111/110,110/100,100/000}
\path[->] (\from) edge[draw = myred, thick]  (\to);
\coordinate (c) at  ($(100.center) + (1mm,0)$);
\path[draw = myred, thick, rounded corners = 5pt,->] (111.center) ..  controls (1,-6) and (-1,-6) ..  (111.center) -- (110.center)  -- (100.center) -- (000.center) ..  controls (-1,1) and (1,1) .. (000.center) -- (001.center) -- (010.center) -- (c) -- (001.center) -- (011.center) -- cycle;
\path[->,shorten >=3mm, shorten <=2mm] (010.center) edge[draw = myred, thick] (c);
\path[->,shorten >=.7cm, shorten <=.9cm] (1,-5.78) edge[draw = myred, thick] (-1,-5.78);
\path[->,shorten >=.7cm, shorten <=.9cm] (-1,.78) edge[draw = myred, thick] (1,.78);
\end{tikzpicture}
\end{center}
\caption{The circuit representing the word $10011110000$ in $G_3^2$}
\label{path}
\end{figure}
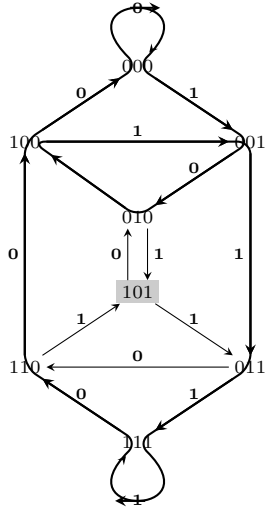

\begin{lem}\label{connect}
	Let $H'$ be an Eulerian subgraph of $G_n^k$ in which each vertex of $G_n^k$ has degree at least two. Then there exists a \emph{connected} Eulerian subgraph $H$ of $G_n^k$ in which each vertex has the same degree as in $H'$. In particular, the number of edges in $H$ is the same as in $H'$.  
\end{lem}	
\begin{proof}
	Suppose that $H'$ is not connected and proceed by induction on the number of its connected components. There exist vertices $at$ and $tb$ in $G_n^k$, where $a$ and $b$ are letters, such that $at \in C_1$ and $tb \in C_2$, where $C_1$ and $C_2$ are  two distinct connected components of $H'$. Let $(at,tc)$ be an edge in $C_1$ and $(dt,tb)$ be an edge in $C_2$. Define $H'_1$ by replacing edges  $(at,tc)$ and $(dt,tb)$ in $H'$ with edges $(at,tb)$ and $(dt,tc)$. The graph $H'_1$ satisfies the hypothesis of the lemma and has 
	a strictly smaller number of connected components.
	Moreover, the degrees of all vertices are not affected by the exchange of edges. This completes the proof.
\end{proof}

We can now reprove \cite[Theorem 1]{Lempel:1971} (see also \cite[Theorem A and Theorem B]{Yoeli:1962}) in the form suitable for our purposes. 
\begin{thm}
	Let $k\geq 2$ and $n\geq 1$. Then for every $N$, $0 < N \leq k^{n+1}$, the graph $G_n^k$ contains a connected Eulerian graph $H$, with $N$ edges and $\min\{k^n,N\}$ vertices. In other words, $H$ is a cycle if $N\leq k^n$, and $H$ contains all vertices of $G_n^k$ otherwise.
\end{thm}
\begin{proof}
	We proceed by induction on $n$. Let first $0 < N \leq k^n$. 
	If $n=1$, then $G_1^k$ contains a cycle of length $N$, since $G_1^k$ is the clique on $k$ vertices (with loops).
	If $n > 1$, then, by the induction hypothesis, the graph $G_{n-1}^k$ contains a circuit of length $N$, which corresponds to a cycle of length $N$ in $G_{n}^k$.
	
	Let now $N = jk^n + N'$ where $1 \leq j \leq k-1$ and $0 < N' \leq k^n$. Let $C$ be a cycle in $G_{n}^k$ of length $N'$ obtained in the previous paragraph, and let $F_1 = \{C,C_1,\cdots,C_m\}$ be a factor of $G_{n}^k$ obtained by Lemma \ref{complete}. The complement of $F_1$ is a regular graph of degree $2k-2$, whose edges can be partitioned into $k-1$ factors $F_2, F_3, \dots, F_k$ by Lemma \ref{petersen}. The edges of $C$, $F_2$, $F_3$, \dots $F_{j+1}$ together yield an Eulerian graph $H'$ with $N$ edges. Each vertex of $G_n^k$ has degree at least two in $H'$. We obtain $H$ from $H'$ using Lemma \ref{connect}.
\end{proof}

We therefore have proved the desired result.
\begin{cor}\label{lempel}
For all integers $k \geq 2$ and $N \geq 1$ there exists a generalized
de Bruijn word of length $N$ over a $k$-letter alphabet.
\end{cor}

\begin{rem}
	We have not been able to find this precise notion of generalized de Bruijn word 
	in the literature anywhere, although there are some papers that come
	very close.   For example, Iv\'anyi \cite{Ivanyi:1987}  considered
	the analogue of the upper bound \eqref{bnd} for ordinary (non-circular) words.  He called a word $w$
	{\it supercomplex} if the bound is attained not only for $w$, but also for all prefixes of $w$.  However, binary supercomplex words do not exist past length $9$.
	The third author also considered the analogue of the bound \eqref{bnd} for ordinary words \cite{Shallit:1993}.
	However, Lemma 3 of that paper actually implies the existence
	of our generalized (circular) de Bruijn words of every length over a binary
	alphabet, although this was not stated explicitly.
	Anisiu, Bl{\'a}zsik, and K{\'a}sa \cite{Anisiu&Blazsik&Kasa:2002} 
	discussed a related concept:  namely, those length-$N$ words $w$ for which
	$\max_{1 \leq i \leq N} \rho_i (w) =
	\max_{x \in \Sigma_k^N} \max_{1 \leq i \leq N} \rho_i (x)$ where $\rho_i (w)$ denotes
	the number of distinct length-$i$ factors of $w$ (here considered in the
	ordinary sense, not circularly).    
	Also see \cite{Flaxman&Harrow&Sorkin:2004}.
\end{rem}

\section{State complexity}
In this section we show that a generalized de Bruijn word can be characterized as a word $w$ with the maximum state complexity of $C(w)$.  
To this end, we first consider a more general setting of languages $L\subseteq \Sigma^N$.
In other words, $L$ is a language containing some words of length $N$ only.

\begin{table}[h]
\begin{center}
\begin{tabular}{p{2ex}|l|c}
$N$ & 
	\begin{tabular}{l}
		lexicographically least generalized \\ binary de Bruijn word of length $N$
	\end{tabular}  & 
	\begin{tabular}{l}
	number of \\ such words
\end{tabular}  
\\
\hline\hline
  1 & 0 & 2 \\ \hline
  2 & 01 & 1 \\ \hline
  3 & 001 & 2 \\ \hline
  4 & 0011 & 1\\ \hline
  5 & 00011 & 2 \\ \hline
  6 & 000111 & 3 \\ \hline
  7 & 0001011 & 4 \\ \hline
  8 & 00010111 & 2 \\ \hline
  9 & 000010111 & 4 \\ \hline
 10 & 0000101111 & 3 \\ \hline
 11 & 00001011101 & 6 \\ \hline
 12 & 000010100111 & 13 \\ \hline
 13 & 0000100110111 & 12 \\ \hline
 14 & 00001001101111 & 20\\ \hline
 15 & 000010011010111 & 32 \\ \hline
 16 & 0000100110101111 & 16 \\ \hline
 17 & 00000100110101111 & 32 \\ \hline
 18 & 000001001101011111 & 36 \\ \hline
 19 & 0000010100110101111 & 68 \\ \hline
 20 & 00000100101100111101 & 141 \\ \hline
 21 & 000001000110100101111 & 242 \\ \hline
 22 & 0000010001101001011111 & 407 \\ \hline
 23 & 00000100011001110101111 & 600 \\ \hline
 24 & 000001000110010101101111 & 898 \\ \hline
 25 & 0000010001100101011011111 & 1440 \\ \hline
 26 & 00000100011001010011101111 & 1812 \\ \hline
 27 & 000001000110010100111011111 & 2000 \\ \hline
 28 & 0000010001100101001110101111 & 2480 \\ \hline
 29 & 00000100011001010011101011111 & 2176 \\ \hline
 30 & 000001000110010110100111011111 & 2816 \\ \hline
 31 & 0000010001100101001110101101111 & 4096
\end{tabular}
\end{center}
\caption{Generalized de Bruijn words}
\label{tab1}
\end{table}

The following theorem determines the maximum state complexity
of such a language for sufficiently large $N$, and characterize languages that achieve the maximum.
Let $\pi_i(L)$ (resp., $\sigma_i(L)$) denote the number of prefixes 
(resp., suffixes) of length $i$ of the language $L$.

\begin{thm}\label{thm1}
	Let $\Sigma$ be an alphabet of cardinality $k\geq 2$, let $N \geq 1$ be an integer, and let
	$L \subseteq \Sigma^N$. Define $m = |L|$ and $r = \left\lfloor \log_k \abs L \right\rfloor$ and $v =1 + k + k^2 + \cdots + k^r$.
	If $N \geq 3r+1$, then  
	\begin{equation}
	\scc(L) \leq 2v  + m\cdot (N-2r-1) + 1.
	\label{maineq}
	\end{equation}
	If $N > 3r + 1$, then equality holds in \eqref{maineq} if
	and only if both of the following two conditions are satisfied:
	\begin{itemize}
		\item[(a)] $\sigma_r(L) = \pi_r(L) = k^r$ 
		\item[(b)] $\sigma_{r+1}(L) = \pi_{r+1}(L) = m$. 
	\end{itemize}
\end{thm}
\begin{proof}
	We use the standard construction of the minimal automaton $\MA$ accepting $L$ as follows. The states $S_\MA$ of $\MA$ are left quotients $\mn p$ of the language $L$, where 
	\[ \mn p = \{ s \mid ps\in L \}\,. \]
	Note that all elements in the state $\mn p$ have the same length $N - |p|$.
	We divide the states of $\MA$ into subsets according to the length of words they contain, as follows:
	\[ S_\MA = A \cup M \cup \bigcup_{\ell=1}^{r} T_\ell \cup \{f\} \cup \{\emptyset\} \]
	where
	\begin{itemize}
		\item $A = \{ \mn p \mid |p|\leq r \}$,
		\item $M = \{ \mn p \mid r < |p| < N - r\}$ ,
		\item $T_\ell = \{ \mn p \mid |p| = N - \ell \}$,
		\item $f= \{\varepsilon\}$.
	\end{itemize}
	The state $f$ is the accepting state, and $\emptyset$ is the ``dead'' state.
	For the size of $A$ we have a bound $v = 1 + k + k^2 + \cdots + k^r$, since $v$ is the number of words $p$ that can define a state $\mn p$.
	
	Let $d = m - \pi_{N-r-1}$. For each length $r< \ell < N-r$, there are at most $\pi_{N-r-1}$ words $p$ of length $\ell$ 
	such that $\mn p$ is nonempty---namely, the prefixes of $L$ of length $\ell$.  Therefore the size of $M$ is at most $(m-d)\cdot (N - 2r - 1)$. 
	
	For $T_\ell$, $1 \leq \ell \leq r$, we need a more detailed analysis, which exhibits a trade-off between the size of $T_\ell$ and the size of $M$. More precisely, we shall show that large $T_\ell$ implies large $d$. Consider the set $T_\ell$ for some fixed $1 \leq \ell \leq r$. Every state $\mn p\in T_\ell$ is a set of words of length $\ell$. ``Expected'' elements of $T_\ell$ are singletons $\{s\}$, with $|s|=\ell$, which yields an ``expected'' size $k^\ell$ of $T_\ell$. Assume that $T_\ell$ contains a state $\mn p$ with cardinality $d_p$ larger than one, say $\mn p = \{s_1,s_2,\dots, s_{d_p}\}$. Then $L$ contains words $ps_1,ps_2,\dots,ps_{d_p}$, all having the same prefix of length $N-r - 1$. This implies that $d$ is at least $d_p-1$. Moreover, the contribution to $d$ is cumulative. Indeed, assume that $\mn {p'} =  \{s_1',s_2',\dots, s_{d_{p'}}'\}$ with $d_{p'}>1$ for some $\mn p \neq \mn {p'} \in T_\ell$. Then $p's_1',p's_2',\dots,p's'_{d_{p'}}$ are pairwise distinct words in $L$ with the same prefix of length $N-r - 1$, and they are also all distinct from any $ps \in L$.  Altogether we have (still with a fixed $\ell$)
	\[d \geq \sum_{\mn p\in T_\ell}(d_p -1),\]
	and the size of $T_\ell$ is at most $k^\ell + d$. Therefore, the set $T = \bigcup_{\ell = 1}^r T_\ell$ has size at most $k + k^2 + \dots + k^r + dr$.
	
	We have shown that 
	\begin{align}\label{d_bound}
	\begin{split}
	\scc(L) &\leq v + (m-d)(N-2r-1) + (v-1 + dr) + 2 = \\
	& = 2v + m(N - 2r -1) + 1 - d(N-3r -1),
	\end{split}
	\end{align}   
	which proves the bound, due to the assumption $N \geq 3r + 1$.
	\medskip
	
	To show the second half of the theorem, note that \eqref{d_bound} and $N > 3r + 1$ imply $d = 0$ if the equality holds in \eqref{maineq}. Therefore states in $T$ are all singletons, and all bounds in the above description have to be achieved. Then the automaton has the topology depicted in Figure \ref{topology} and the two conditions are satisfied.
	\begin{figure}[H]
		\begin{center}
			\begin{tikzpicture}
[baseline,->,>=stealth,
state/.style={circle,inner sep = 0, minimum size=2mm,very thin,draw=black,initial text=},
every node/.style={font=\small}]
\def \step {.55}
\def \dist {.3}
\draw[white, fill=gray!50] (-.5*\step,-8*\dist) rectangle (3.5*\step,8*\dist);
\draw[white, fill=gray!50] (11.5*\step,-8*\dist) rectangle (15.5*\step,8*\dist);
\draw[white] (3.5*\step,-8*\dist) rectangle (11.5*\step,8*\dist);
\node[state,initial]  (e) {};
\node[state]   (0) at (\step,4*\dist) {};
\node[state]  (1) at (\step,-4*\dist) {};
\foreach \y in {0,...,3}
\node[state]  (1\y) at (2*\step,6*\dist-4*\dist*\y) {};
\foreach \y in {0,...,7}
\node[state]  (2\y) at (3*\step,7*\dist-2*\dist*\y) {};
\foreach \from/\to in {e/0,e/1,0/10,0/11,1/12,1/13,10/20,10/21,11/22,11/23,12/24,12/25,13/26,13/27}
\path (\from) edge (\to);
\foreach \x in {4,...,11}
\foreach \y in {0,2,3,5,7,8,9,11,13,15}
\node[state]  (\x\y) at (\x*\step,7.5*\dist-\dist*\y) {};
\foreach \from/\to in {0/0,1/2,1/3,2/5,3/7,4/8,4/9,5/11,6/13,7/15}
\path (2\from) edge (4\to);
\foreach \x/\xx in {4/5,5/6,6/7,7/8,8/9,9/10,10/11}
\foreach \y in {0,2,3,5,7,8,9,11,13,15}
\path (\x\y) edge (\xx\y);
\node[state,accepting,inner sep = 2pt] at (15*\step,0) (f) {$f$};
\node[state]   (f0) at (14*\step,4*\dist) {};
\node[state]  (f1) at (14*\step,-4*\dist) {};
\foreach \y in {0,...,3}
\node[state]  (f1\y) at (13*\step,6*\dist-4*\dist*\y) {};
\foreach \y in {0,...,7}
\node[state]  (f2\y) at (12*\step,7*\dist-2*\dist*\y) {};
\foreach \to/\from in {/0,/1,0/10,0/11,1/12,1/13,10/20,10/21,11/22,11/23,12/24,12/25,13/26,13/27}
\path (f\from) edge (f\to);
%
\foreach \from/\to in  {1/0,1/2,0/3,2/5,2/7,5/8,3/9,4/11,7/13,6/15}
\path (11\to) edge (f2\from);
\draw [-,decorate,decoration={brace,amplitude=10pt},xshift=-0pt,yshift=0pt]
(-.5*\step,8*\dist) -- (3.5*\step,8*\dist) node [black,midway,yshift=0.6cm] 
{\footnotesize $r = 3$};
\draw [-,decorate,decoration={brace,amplitude=10pt},xshift=-0pt,yshift=0pt]
(11.5*\step,8*\dist) -- (15.5*\step,8*\dist) node [black,midway,yshift=0.6cm] 
{\footnotesize $r$};
\draw [-,decorate,decoration={brace,amplitude=10pt},xshift=-0pt,yshift=0pt]
(3.5*\step,8*\dist) -- (11.5*\step,8*\dist) node [black,midway,yshift=0.6cm] 
{\footnotesize $N - 2r -1$};
\draw [-,decorate,decoration={brace,amplitude=10pt},xshift=-0pt,yshift=0pt]
(15.5*\step,8*\dist) -- (15.5*\step,-8*\dist) node [black,midway,xshift=1cm] 
{\footnotesize $m = 10$};
\foreach \from/\to in  {e/1,1/12,12/24,24/48}
\path (\from) edge (\to);
\foreach \from/\to in  {4/5,5/6,6/7,7/8,8/9,9/10,10/11}
\path (\from8) edge (\to8);
\foreach \from/\to in  {118/f25,f25/f12,f12/f1,f1/f}
\path (\from) edge (\to);
\node at (1.8*\step,-9*\dist) {$A$};
\node at (7*\step,-9*\dist) {$M$};
\node at (13*\step,-9*\dist) {$T$};
\end{tikzpicture}
		\end{center}
		\caption{Example of the maximum automaton topology}
		\label{topology}
	\end{figure}
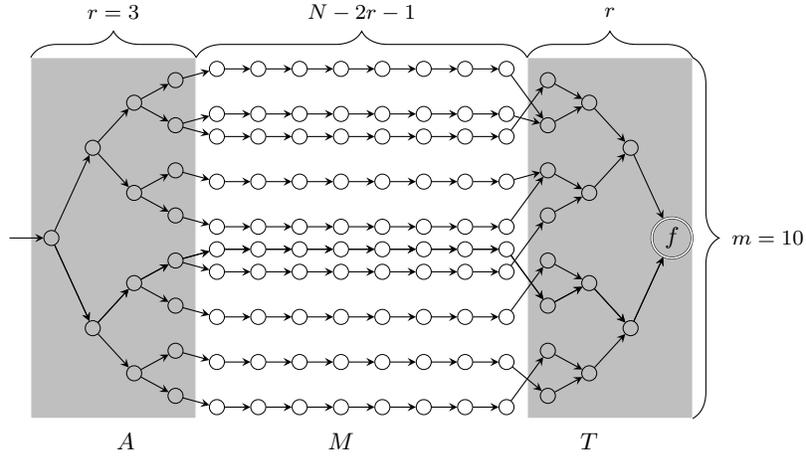
	Now assume that conditions (a) and (b) of the theorem are satisfied. Let $p$ be a prefix of a word in $L$ with $|p|> r$, and assume that $s_1, s_2\in \mn p$ for two distinct words $s_1$ and $s_2$. Then $ps_1,ps_2\in L$ have the same prefix of length $r+1$, a contradiction with $\pi_{r+1}=m$. Therefore all $\mn p$ in $M\cup T$ are singletons. From $\sigma_r = k^r$ we now deduce that $T_\ell = \{ \{s\} |\, \ell = |s|\}$ for each $\ell = 1,2,\dots,r$, and $T$ has size $k + k^2 + \cdots + k^{r}$. 
	
	Let $p_1$ and $p_2$ be two distinct prefixes in $L$ of length at most $N-r-1$ such that some $s$ is in both $\mn {p_1}$ and $\mn {p_2}$, which are states in $A\cup M$. Then $p_1s$ and $p_2s$ are two distinct words in $L$ with the same suffix of length $r+1$, a contradiction with $\sigma_{r+1}=m$. Therefore states in $A\cup M$ are pairwise disjoint. From $\pi_{r+1}=m$ we deduce that $L$ has $m$ distinct prefixes for each size $r< \ell < N - r$, hence the size of $M$ is $m\cdot (N - 2r - 1)$. Finally, from $\pi_r=k^r$ we obtain that $A$ contains $v$ distinct states.
	The ``dead'' state $\emptyset$ completes the bound. 
\end{proof} 

In the conference version of our paper we mistakenly claimed that 
Theorem~\ref{thm1} holds for $N \geq 2r + 1$ instead of  $N \geq 3r + 1$.
The following example shows that this claim was incorrect, and that the bound $N \geq 3r + 1$ is optimal. Consider the language
\[
L = \{000000,000001,010000,100010,110101,111011\}.
\]
We have $m=6$, $r=2$ and $N = 3r = 6$. The state complexity of $L$ is $22$ while $2v + m(N - 2r -1) + 1 = 21$. The minimal automaton for $L$ is shown in Figure \ref{counter}. Compared to the topology of Figure \ref{topology}, there is one state missing in part $M$ ($d=1$) which allows two non-singleton states in $T_2$ and $T_1$ (the ``dead'' state is not shown).  
\begin{figure}[H]
	\begin{center}
		\begin{tikzpicture}
[baseline,->,>=stealth,
state/.style={circle,inner sep = 2pt, draw = white, minimum size=5mm,very thin, initial text=},
every node/.style={font=\small}]
\def \step {1.2}
\def \dist {.3}
\node[state,initial]  (e) {$L$};
\node[state]   (0) at (\step,4*\dist) {$\mn 0$};
\node[state]  (1) at (\step,-4*\dist) {$\mn 1$};
\foreach \y/\p in {0/00,1/01,2/10,3/11}
\node[state]  (1\y) at (2*\step,6*\dist-4*\dist*\y) {$\mn\p$};
\foreach \y/\p/\d in {0/000/4,1/010/4,2/100/4,3/110/3.5,4/111/3.5}
\node[state]  (2\y) at (3.5*\step,6*\dist-\d*\dist*\y)  {$\mn\p$};
\foreach \y/\p/\d in {0/{00,01}/2,1/00/3,2/10/3.5,3/01/3.65,4/11/3.7}
\node[state]  (3\y) at (5*\step,9*\dist-\d*\dist*\y)  {$\{\p\}$};
\node[state]  (f) at (8*\step,0.1) {$\{\varepsilon\}$};
\node[state]   (0f) at (6.5*\step,4*\dist) {$\{0\}$};
\node[state]  (1f) at (6.5*\step,-4*\dist) {$\{1\}$};
\node[state]  (ef) at (6.5*\step,7*\dist) {$\{0,1\}$};
\node[rectangle, rounded corners = 5pt, fill=gray!50] at (5*\step,9*\dist) {$\{00,01\}$}; 
\node[rectangle, rounded corners = 5pt, fill=gray!50] at (6.5*\step,7*\dist) {$\{0,1\}$};
\node[rectangle, rounded corners = 5pt, fill=gray!50, color=gray!50] at (3.5*\step,9*\dist) {$[]000$};  
\foreach \from/\to/\pis in {e/0/0,e/1/1,0/10/0,0/11/1,1/12/0,1/13/1,10/20/0,11/21/0,12/22/0,13/23/0,13/24/1,20/30/0,21/31/0,22/32/0,23/33/1,24/34/0,30/ef/0,31/0f/0,32/0f/1,33/1f/0,34/1f/1,0f/f/0,1f/f/1,ef/f/{0,1}}
\path (\from) edge node[fill = white, inner sep = 1pt] {\tiny \pis} (\to);
\end{tikzpicture}
	\end{center}
	\caption{A counter-example for $N=3r$}
	\label{counter}
\end{figure}
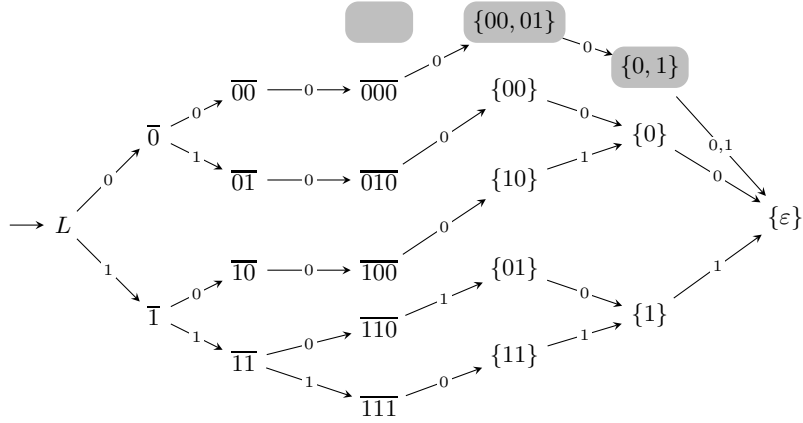
The slightly modified language
\[L' =  \{0000000,0000001,0100000,1000010,1100101,1110011\}\]
also shows that for $N=3r+1$, the maximum can be achieved with a different topology, namely with $\pi_{r+1}=\sigma_{r+1}=m-1$, see Figure \ref{counter2}.
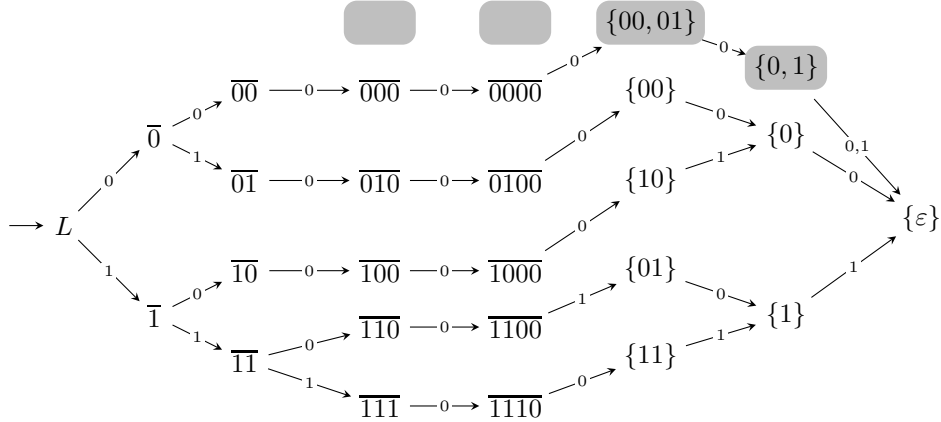
\begin{figure}[H]
	\begin{center}
		\begin{tikzpicture}
[baseline,->,>=stealth,
state/.style={circle,inner sep = 2pt, draw = white, minimum size=5mm,very thin, initial text=},
pismeno/.style={fill=white}
every node/.style={font=\small}]
\def \step {1.2}
\def \dist {.3}
\node[state,initial]  (e) {$L$};
\node[state]   (0) at (\step,4*\dist) {$\mn 0$};
\node[state]  (1) at (\step,-4*\dist) {$\mn 1$};
\foreach \y/\p in {0/00,1/01,2/10,3/11}
\node[state]  (1\y) at (2*\step,6*\dist-4*\dist*\y) {$\mn \p$};
\foreach \y/\p/\d in {0/000/4,1/010/4,2/100/4,3/110/3.5,4/111/3.5}
\node[state]  (2\y) at (3.5*\step,6*\dist-\d*\dist*\y)  {$\mn \p$};
\foreach \y/\p/\d in {0/0000/4,1/0100/4,2/1000/4,3/1100/3.5,4/1110/3.5}
\node[state]  (2a\y) at (5*\step,6*\dist-\d*\dist*\y)  {$\mn \p$};
\foreach \y/\p/\d in {0/{00,01}/2,1/00/3,2/10/3.5,3/01/3.65,4/11/3.7}
\node[state]  (3\y) at (6.5*\step,9*\dist-\d*\dist*\y)  {$\{\p\}$};
\node[state]  (f) at (9.5*\step,0.1) {$\{\varepsilon\}$};
\node[state]   (0f) at (8*\step,4*\dist) {$\{0\}$};
\node[state]  (1f) at (8*\step,-4*\dist) {$\{1\}$};
\node[state]  (ef) at (8*\step,7*\dist) {$\{0,1\}$};
\node[rectangle, rounded corners = 5pt, fill=gray!50] at (6.5*\step,9*\dist) {$\{00,01\}$}; 
\node[rectangle, rounded corners = 5pt, fill=gray!50] at (8*\step,7*\dist) {$\{0,1\}$};
\node[rectangle, rounded corners = 5pt, fill=gray!50, color=gray!50] at (3.5*\step,9*\dist) {$[]000$};  
\node[rectangle, rounded corners = 5pt, fill=gray!50, color=gray!50] at (5*\step,9*\dist) {$[]000$};  
\foreach \from/\to/\pis in {e/0/0,e/1/1,0/10/0,0/11/1,1/12/0,1/13/1,10/20/0,11/21/0,12/22/0,13/23/0,13/24/1,2a0/30/0,2a1/31/0,2a2/32/0,2a3/33/1,2a4/34/0,30/ef/0,31/0f/0,32/0f/1,33/1f/0,34/1f/1,0f/f/0,1f/f/1,ef/f/{0,1},20/2a0/0,21/2a1/0,22/2a2/0,23/2a3/0,24/2a4/0}
\path (\from) edge node[fill = white, inner sep = 1pt] {\tiny \pis} (\to);
\end{tikzpicture}
	\end{center}
	\caption{A counter-example for $N=3r+1$}
	\label{counter2}
\end{figure}

We can now formulate our result on the state complexity of generalized de Bruijn words.

\begin{thm}\label{thm_complexity}
If $w$ is a word of length $N$ over a $k$-letter
alphabet, with $k\geq 2$, then 
\begin{align}\label{maineq2}
\scc(C(w)) \leq 2v + N(N-2r-1) + 1,
\end{align}
where $r = \lfloor \log_k N \rfloor$ and $v = 1 + k + k^2 + \cdots + k^r$. 

Moreover, equality holds in \eqref{maineq2} iff $w$ is a generalized de Bruijn word.
\end{thm}

\begin{proof}
Let $w$ be a word of length $N$, and let $L = C(w)$. Note that, for each $1 \leq i \leq N$, we have $\pi_i(L) = \sigma_i(L) = \gamma_i(w)$.
Therefore, the theorem follows from Theorem \ref{thm1} if $N > 3r + 1$.
 
 For $N \leq 3r + 1$, the claim has to be checked separately. This concerns the following cases:
 \begin{itemize}
 	\item $N = 1$ for any $k \geq 2$;
 	\item $1 < N \leq 10$ for $k = 2$;
 	\item $N = 3$ and $N = 4$ for $k = 3$; and
 	\item $N = 4$ for $k = 4$.
 \end{itemize}
If $|w|=1$, then $r=0$, $v=1$, and the minimal accepting automaton has three states: $\{w\}$, $\{\varepsilon\}$ and $\emptyset$. Moreover, $w$ is a generalized de Bruijn word, since $\gamma_0(w) = \gamma_1(w) = k^0 = N = 1$. Therefore the theorem holds in this case.

Table \ref{search} lists all generalized de Bruijn words (up to the conjugation and the exchange of letters) for the remaining cases not covered by Theorem \ref{thm1}. We verified by an exhaustive computer search
that they are exactly the words for which equality holds in \eqref{maineq2}, and that no other word has a larger complexity.   
\end{proof}
\begin{table}[H]
	\begin{center}
		\begin{tabular}{c|c|l}
			$k$ & $N$ & maximum words \\
			\hline
			2 & 2 & 01 \\
			2 & 3 & 001 \\
			2 & 4 & 0011 \\
			2 & 5 & 00011\\
			2 & 6 & 000111, 001011 \\
			2 & 7 & 0001011, 0001101 \\
			2 & 8 & 00010111 \\
			2 & 9 & 000010111, 000011101 \\
			2 & 10 & 0000101111, 0001011101 \\
			3 & 3 & 012\\
			3 & 4 & 0012,0102 \\
			4 & 4 & 0123
		\end{tabular}
	\end{center}
	\caption{Maximum words not covered by Theorem \ref{thm1}}
	\label{search}
\end{table}

For $k = 2$ the maximum state complexity of $C(x)$ over length-$N$ words $x$ is given in Table~\ref{tab2} for $1 \leq N \leq 10$.   It is sequence
\seqnum{A316936} in the OEIS \cite{Sloane}.
\begin{table}[H]
\begin{center}
\begin{tabular}{c|c}
$N$ & $\max_{x\in\Sigma_2^N} \scc(C(x))$ \\
\hline
1 & 3 \\
2 & 5 \\
3 & 7 \\
4 & 11\\
5 & 15 \\
6 & 21 \\
7 & 29 \\
8 & 39 \\
9 & 49 \\
10 & 61
\end{tabular}
\end{center}
\caption{Maximum state complexity of conjugates of binary words of length $N$}
\label{tab2}
\end{table}

\section{Counting generalized de Bruijn words}
We first count the total number of factors of a generalized de Bruijn word.
This is a generalization of Theorem 2 of \cite{Shallit:1993} to all
$k \geq 2$, adapted for the case of circular words.

\begin{prop}
	If $w \in \Sigma_k^N$ is a generalized de Bruijn word, then
	$$ \sum_{0 \leq i \leq N} \gamma_i (w) = {{k^{r+1} -1 } \over {k-1}} + N(N-r),$$
	where $r = \lfloor \log_k N \rfloor$.
\end{prop}

\begin{proof}
	We have
	\begin{align*}
	\sum_{0 \leq i \leq N} \gamma_i (w) &=
	\sum_{0 \leq i\leq N}  \min(k^i,N)  \\
	& = \sum_{0 \leq i \leq r} k^i + \sum_{r < i \leq N} N \\
	& = {{k^{r+1} -1 } \over {k-1}} + N(N-r).
	\end{align*}
\end{proof}

Counting the exact
number of generalized de Bruijn words of length $N$ appears
to be a difficult task.
Figures for small $N$ can be obtained by a computer search,
as in Table \ref{tab1}.
The second author has computed these numbers up to $N = 64$ (see Table \ref{tab63} for a possible independent verification). 
\begin{table}[H]
	\begin{center}
		\begin{tabular}{c c c c c}
			\begin{tabular}{c|r}
				length & number
				\\
				\hline
				32 & 2 048 \\
				33 & 4 096 \\
				34 & 3 840 \\
				35 & 7 040 \\
				36 & 13 744 \\
				37  & 28 272 \\
				38  & 54 196 \\
				39  & 88 608 \\
				40  & 160 082 \\
				41  & 295 624 \\
				42  & 553 395 \\
			\end{tabular}
			&
			&
			\begin{tabular}{c|r}
				length & number
				\\
				\hline
				43  & 940 878 \\
				44  & 1 457 197 \\
				45  & 2 234 864 \\
				46  & 3 302 752 \\
				47  & 4 975 168 \\
				48  & 7 459 376 \\
				49  & 10 347 648 \\
				50  & 13 841 408 \\
				51  & 17 696 256 \\
				52  & 23 404 848 \\
				53  & 30 918 336 \\
			\end{tabular}
			&
			&
			\begin{tabular}{c|r}
				length & number
				\\
				\hline
				54  & 36 137 280 \\
				55  & 38 730 752 \\
				56  & 41 246 208 \\
				57  & 50 774 016 \\
				58  & 60 764 160 \\
				59  & 62 619 648 \\
				60  & 70 057 984 \\
				61  & 59 768 832 \\
				62  & 88 080 384 \\
				63  & 134 217 728 \\
				64  & 268 435 456  
			\end{tabular}
		\end{tabular}
	\end{center}
	\caption{Numbers of longer binary generalized de Bruijn words}
	\label{tab63}
\end{table}
Except in a few simple cases, we do not even know an exact asymptotic expression.
For example, if $N = k^n$, then it follows from known results
\cite{Aardenne-Ehrenfest&de.Bruijn:1951} that this
number is $(k!)^{k^{n-1}}/k^n$, counted up to cyclic shift.
Some loose bounds could be obtained from \cite{Maurer:1992}, keeping in mind, however, that we are interested in circuits visiting all vertices, not just arbitrary circuits.
Precise numbers seem to be relatively easily computable for $N = k^n \pm 1$, and possibly also for $N = k^n \pm 2$. In particular, the number of binary generalized de Bruijn words of length $N = 2^n \pm 1$ is twice the number of such words of length $2^n$; see the discussion in \cite[p.~202]{Fredricksen:1982}. The considerations, however, quickly become involved. It can be verified by a computer search, for example, that the formula for cycles of length  $2^n -2$ given in \cite[p.~203]{Fredricksen:1982} is wrong. Similarly, the number of cycles of length $k^n\pm 1$ we gave in the final comments of our conference paper is also wrong for $k > 2$. For example, computer search shows that the number of ternary generalized de Bruijn words of lengths 8 and 10 are 36 and 108, respectively, while the number of ternary (generalized) de Bruijn words of length 9 is 24. We therefore leave this question open for further research.  

\section*{Acknowledgments}

We thank the anonymous referees for helpful comments and suggestions.

\bibliographystyle{elsarticle-num}
\bibliography{abbrevs,conj2}

\end{document}